\documentclass[11pt]{article}

\usepackage{amsmath,amssymb,amsthm,graphicx,booktabs}
\usepackage[margin=1.2in]{geometry}
\usepackage[round]{natbib}
\usepackage{setspace}
\newcommand{\R}{\mathbb{R}}
\newcommand{\tr}{\operatorname{tr}}
\newcommand{\diag}{\operatorname{diag}}
\newcommand{\vecl}{\operatorname{vecl}}
\newcommand{\one}{\mathbf{1}}
\newcommand{\Sym}{\mathbb{S}_n}
\newcommand{\SPD}{\mathbb{S}_n^{++}}
\newcommand{\Cn}{\mathcal{C}_n}

\newtheorem{theorem}{Theorem}
\newtheorem{corollary}{Corollary}
\newtheorem{proposition}{Proposition}

\title{Fast inversion of the generalized Fisher transformation
of correlation matrices}

\author{Ilya Archakov\thanks{Department of Economics, York University,
Toronto, Ontario M3J 1P3, Canada; iarch@yorku.ca.}
\and Peter Reinhard Hansen\thanks{Department of Economics, University
of North Carolina at Chapel Hill, Chapel Hill, North Carolina 27599,
U.S.A.; hansen@unc.edu.}}

\date{September 2026}

\begin{document}
\maketitle

\begin{abstract}
\noindent The generalized Fisher transformation maps a non-singular correlation matrix to an unconstrained real vector through the off-diagonal elements of its matrix logarithm. Evaluating its inverse is a computational bottleneck in dynamic correlation and multivariate volatility models. We develop a fast inversion algorithm by characterizing the unknown diagonal as the minimizer of a smooth, strictly convex, and coercive objective. An explicit Hessian and global spectral bounds identify the standard fixed-point iteration as a quasi-Newton method and explain why it can converge slowly near singularity. Every fixed-point step decreases the objective, and the iteration converges from every starting point. These results motivate GFT-FP+N, a hybrid of fixed-point and matrix-free Newton steps that never forms the Jacobian. In benchmarks with up to $1{,}000$ replications per design and dimensions up to $800$, GFT-FP+N reduces computation time by up to a factor of forty-five relative to the fixed-point iteration and converged in every replication, including on designs where Broyden's method almost always fails. Julia and R packages are provided.

\end{abstract}

\section{Introduction}\label{sec:intro}

The transformation $\gamma(C)=\vecl(\log C)$, where $\log$ is the matrix logarithm and $\vecl$ stacks the below-diagonal elements, maps the set $\Cn$ of non-singular $n\times n$ correlation matrices bijectively onto $\R^{d}$, $d=n(n-1)/2$ \citep{ArchakovHansen:Correlation}. For $n=2$ it reduces to the classical Fisher transformation \citep{Fisher:1921}, and it is accordingly known as the generalized Fisher transformation. Its unconstrained range facilitates dynamic correlation and multivariate stochastic volatility modeling \citep{TongHansenArchakov:2026,ChenFeiYu:2025}, simulation and inference \citep{ArchakovHansenLuo-RandomCorr:2024,ArchakovHansen:2026}, and structured correlation modeling \citep{ArchakovHansen:CanonicalBlockMatrix}. It continues a long line of unconstrained covariance and correlation parametrizations \citep{PinheiroBates:1996,Pourahmadi:1999}.

Many of these applications repeatedly evaluate the inverse mapping, for which no general closed form is available. In a dynamic model, inversion may be required at every time point within each likelihood evaluation, making its computational cost consequential. \citet{ArchakovHansen:Correlation} established the bijection through a contraction argument and proposed a fixed-point iteration that adjusts the diagonal of the matrix logarithm until its exponential has unit diagonal. They derived the iteration's Jacobian, related convergence to its spectral radius, and documented slow convergence for near-singular matrices. \citet{ChenFeiYu:2025} identified inversion as the principal computational bottleneck in multivariate stochastic volatility estimation and compared the fixed point with Newton's and Broyden's methods. Their Newton implementation forms the exact Jacobian at order $n^{4}$ per iteration, whereas Broyden's method forms it once and applies rank-one updates; they adopted Broyden's method.

\citet{Zwiernik:2025} placed the parametrization in a general framework of entropic covariance models, in which selected entries of a covariance matrix are paired with the complementary entries of its matrix logarithm. The variational formulation used here is the unit-diagonal specialization of that framework, with existence and uniqueness following from the Bregman-divergence argument of \citet[Theorem~5.4]{Zwiernik:2025}.

We use this formulation to develop a fast and reliable inversion algorithm. For a prescribed $z\in\R^{d}$, let $A[x]$ be the symmetric matrix with $\vecl(A[x])=z$ and diagonal $x$. Evaluating the inverse reduces to minimizing the smooth, strictly convex, and coercive objective $f(x)=\tr(e^{A[x]})-\one'x$ over $x\in\R^{n}$ (Theorem~\ref{thm:var}). Our direct proof yields an explicit Hessian. Theorem~\ref{prop:spec} synthesizes global spectral bounds, several of which are implicit in the appendix of \citet{ArchakovHansen:Correlation}, and provides a proof through the integral representation of the Hessian. These bounds characterize the local convergence factor of the fixed-point iteration and show that the Newton system is never worse conditioned than $e^{A[x]}$, before or after diagonal preconditioning.

A Golden--Thompson inequality shows that every fixed-point step decreases $f$, yielding a short proof of global convergence (Proposition~\ref{prop:descent}). This descent property provides the basis for combining the fixed point with Newton acceleration. The resulting algorithm, GFT-FP+N, computes Newton steps for the log-diagonal residual by preconditioned conjugate gradients using Hessian-vector products, avoiding explicit Jacobian formation. Fixed-point steps and Newton steps accepted by the Armijo line search decrease the same objective. The implementation also uses evaluation in the log domain and safeguards for floating-point arithmetic, described in Section~\ref{sec:algorithm}.

The experiments in Section~\ref{sec:numerics} assess speed and reliability across designs with dimensions up to $800$ and up to $1{,}000$ replications per design. GFT-FP+N reduces computation time by up to a factor of forty-five relative to the fixed-point iteration and converged in every replication, including on designs where Broyden's method almost always fails. Comparisons with full Newton using the same safeguards isolate the computational benefit of matrix-free linear solves, while comparisons with Anderson acceleration and limited-memory BFGS assess performance against general-purpose alternatives. A sequential inversion experiment examines the gains available with warm starts. Julia and R packages are provided.

\section{A variational characterization}\label{sec:var}

Let $\Sym$ denote the symmetric $n\times n$ matrices, $\SPD$ the
positive definite ones, and $\Cn=\{C\in\SPD:\diag(C)=\one\}$, where
$\diag(\cdot)$ is the vector of diagonal elements and
$\one=(1,\ldots,1)'$. The matrix exponential is a bijection from $\Sym$
onto $\SPD$ \citep[Ch.~11]{Higham:2008}, so $C\mapsto A=\log C$
identifies $\Cn$ with $\{A\in\Sym:\diag(e^{A})=\one\}$. For $z\in\R^{d}$
and $x\in\R^{n}$ let $A[x]=A[x;z]$ be the symmetric matrix with
$\vecl(A[x])=z$ and diagonal $x$.

\begin{theorem}\label{thm:var}
Fix $z\in\R^{d}$ and define $f:\R^{n}\to\R$ by
\begin{equation}\label{eq:f}
f(x)=\tr\big(e^{A[x]}\big)-\one'x .
\end{equation}
Then $f$ is smooth, strictly convex, and coercive, with
\[
\nabla f(x)=\diag(e^{A[x]})-\one,
\qquad
\nabla^{2}f(x)=H(x)=\int_{0}^{1}e^{sA[x]}\circ e^{(1-s)A[x]}ds\succ0,
\]
where $\circ$ is the Hadamard product. Consequently $f$ has a unique
minimizer $x^{\ast}(z)$, which is the unique solution of
$\diag(e^{A[x]})=\one$, and $z\mapsto x^{\ast}(z)$ is real-analytic.
\end{theorem}

\begin{proof}
For $A,D\in\Sym$, differentiating the power series termwise and using
cyclicity of the trace gives
$\frac{d}{dt}\tr e^{A+tD}|_{t=0}=\tr(e^{A}D)$; with $D=e_ie_i'$ this
yields the gradient. For the Hessian, the Fr\'echet derivative of the
exponential is
$\frac{d}{dt}e^{A+tD}|_{t=0}=\int_{0}^{1}e^{sA}De^{(1-s)A}ds$, so with
$D=\diag(u)$,
\begin{equation}\label{eq:quadform}
u'H(x)u=\int_{0}^{1}\tr\big(De^{sA}De^{(1-s)A}\big)ds
=\int_{0}^{1}\sum_{i,j}u_iu_j(e^{sA})_{ij}(e^{(1-s)A})_{ij}ds ,
\end{equation}
identifying $H$. For each $s$ the matrices $e^{sA}$ and $e^{(1-s)A}$ are
positive definite, so their Hadamard product is positive definite by the
Schur product theorem \citep[Theorem~5.2.1]{HornJohnson:1991}; hence
$H(x)\succ0$ and $f$ is strictly convex. For coercivity, let
$A=Q\Lambda Q'$ with $Q=(q_{ik})$ orthonormal; since the weights
$(q_{ik}^{2})_k$ sum to one, Jensen's inequality gives
$(e^{A})_{ii}=\sum_k q_{ik}^{2}e^{\lambda_k}\geq e^{A_{ii}}$, so
$f(x)\geq\sum_i(e^{x_i}-x_i)\to\infty$ as $\|x\|\to\infty$. A coercive
continuous function attains a minimum, which by strict convexity is the
unique stationary point, and by the gradient formula stationarity is
exactly $\diag(e^{A[x]})=\one$. Since $H\succ0$ everywhere, the implicit
function theorem applied to $(x,z)\mapsto\diag(e^{A[x;z]})-\one$ shows
that $x^{\ast}(z)$ is real-analytic.
\end{proof}

\begin{corollary}[\citealp{ArchakovHansen:Correlation}]\label{cor:bijection}
The mapping $\gamma:\Cn\to\R^{d}$, $\gamma(C)=\vecl(\log C)$, is a
bijection, and $\gamma^{-1}(z)=\exp(A[x^{\ast}(z);z])$ is real-analytic.
\end{corollary}

The bijection was first established by
\citet{ArchakovHansen:Correlation} through a contraction argument, and
it is also implied by the mixed-parametrization theory of
\citet{Zwiernik:2025}: for $F(\Sigma)=\tr(\Sigma\log\Sigma-\Sigma)$,
whose gradient is $\log\Sigma$ and whose conjugate is
$F^{\ast}(L)=\tr(e^{L})$, \citet[Theorem~5.4, Remark~5.6, and
Section~5.2]{Zwiernik:2025} obtain the inverse by minimizing the
Bregman divergence $F(S)+F^{\ast}(L)-\langle L,S\rangle$ over
$L=A[x;z]$ for any fixed $S$ with unit diagonal, and since
$\langle A[x;z],S\rangle=\one'x+\mathrm{const}$ this is
\begin{equation}\label{eq:varchar}
x^{\ast}(z)=\arg\min_{x\in\R^{n}}\ \tr\big(e^{A[x;z]}\big)-\one'x ,
\end{equation}
the program of Theorem~\ref{thm:var}. For such problems
\citet[Section~6]{Zwiernik:2025} proposes general-purpose first-order
methods, projected gradient descent and iterative Bregman projection,
noting that second-order information is hard to obtain for most
choices of $F$ and that every step costs a spectral decomposition. For
the unit-diagonal case the direct proof above supplies that
second-order information: the Hessian representation makes
\eqref{eq:varchar} an unconstrained, strictly convex program in $n$
variables with explicit curvature, which the remainder of the paper
exploits.

\section{Global spectral bounds on the Hessian}\label{sec:spectral}

Throughout this section fix $x\in\R^{n}$ and write $E=e^{A[x]}$,
$d=\diag(E)$ and $D=\diag(d)$, the diagonal matrix with the elements of
$d$ on its diagonal; $\operatorname{cond}$ denotes the spectral
condition number, the ratio of extreme eigenvalues. Parts of the
following theorem are implicit in the appendix of
\citet{ArchakovHansen:Correlation}, which obtained the Jacobian
$I-D^{-1}H(x)$ of the fixed point, represented $H(x)$ as a principal
submatrix of $(Q\otimes Q)\Xi(Q\otimes Q)'$ following
\citet{LintonMcCrorie:1995}, and showed in its Lemma~A.3 that
$I-D^{-1/2}H(x)D^{-1/2}$ is positive semidefinite; the row-sum
identity, the lower bound and the diagonal bound can be read off from
those formulas. The theorem synthesizes these facts, proves them
through the integral representation of Theorem~\ref{thm:var}, and
adds the attained upper bound and the condition-number inequalities.
Sharper bounds that depend on the spectral spread of $e^{A[x]}$, and
the duality theory behind the variational formulation, are developed
in \citet{ArchakovHansen:Legendre}.

\begin{theorem}\label{prop:spec}
For every $x\in\R^{n}$,
\begin{itemize}
\item[(i)] $H(x)\one=d$;
\item[(ii)] $\lambda_{\min}(E)I\preceq H(x)\preceq D\preceq\lambda_{\max}(E)I$;
\item[(iii)] $D^{-1/2}H(x)D^{-1/2}\preceq I$, with equality attained:
$D^{-1/2}H(x)D^{-1/2}$ has largest eigenvalue one, with eigenvector
$D^{1/2}\one$; consequently
\[
\operatorname{cond}\{H(x)\}\leq\operatorname{cond}(E),
\qquad
\operatorname{cond}\{D^{-1/2}H(x)D^{-1/2}\}\leq\operatorname{cond}(E).
\]
\end{itemize}
At the solution $x^{\ast}$, where $d=\one$ and $E=C$, these specialize
to $H^{\ast}\one=\one$, $\lambda_{\max}(H^{\ast})=1$, and
$\lambda_{\min}(H^{\ast})\geq\lambda_{\min}(C)$, with
$H^{\ast}=H(x^{\ast})$.
\end{theorem}

\begin{proof}
(i) For symmetric $M,N$, $(M\circ N)\one=\diag(MN)$, so
$H(x)\one=\int_{0}^{1}\diag(e^{sA}e^{(1-s)A})ds=\diag(E)=d$.

(ii) For $u\in\R^{n}$ and $D_u=\diag(u)$, write the integrand of
\eqref{eq:quadform} as
\[
\varphi(s)=\tr(D_ue^{sA}D_ue^{(1-s)A})
=\sum_{k,l}(Q'D_uQ)_{kl}^{2}e^{s\lambda_k+(1-s)\lambda_l},
\]
with $A=Q\Lambda Q'$. Each summand is convex in $s$ and
$\varphi(0)=\varphi(1)=\tr(D_u^{2}E)$, so $\varphi(s)\leq\varphi(0)$ on
$[0,1]$ and
\[
u'H(x)u=\int_{0}^{1}\varphi(s)ds\leq\tr(D_u^{2}E)
=\sum_i u_i^{2}d_i=u'Du ,
\]
which is $H(x)\preceq D$; and $D\preceq\lambda_{\max}(E)I$ because
$d_i=e_i'Ee_i\leq\lambda_{\max}(E)$. For the lower bound,
$\varphi(s)=\tr\{(D_ue^{sA}D_u)e^{(1-s)A}\}
\geq\lambda_{\min}(e^{(1-s)A})\tr(D_u^{2}e^{sA})
\geq e^{(1-s)\lambda_{\min}}e^{s\lambda_{\min}}\|u\|^{2}
=\lambda_{\min}(E)\|u\|^{2}$, with
$\lambda_{\min}=\lambda_{\min}(A)$.

(iii) By (ii), $D^{-1/2}H(x)D^{-1/2}\preceq I$, and by (i),
$D^{-1/2}H(x)D^{-1/2}(D^{1/2}\one)=D^{-1/2}H(x)\one=D^{-1/2}d
=D^{1/2}\one$, so one is an eigenvalue, necessarily the largest. Both
condition-number bounds combine the two sides of (ii), the second using
$D^{-1/2}H(x)D^{-1/2}\succeq\lambda_{\min}(E)D^{-1}\succeq
\{\lambda_{\min}(E)/\lambda_{\max}(E)\}I$.
\end{proof}

Theorem~\ref{prop:spec} has two consequences for the minimization of
\eqref{eq:f}. First, consider the fixed-point iteration of
\citet{ArchakovHansen:Correlation}, hereafter GFT-FP,
\begin{equation}\label{eq:fp}
x\ \leftarrow\ x-\log\diag\big(e^{A[x]}\big).
\end{equation}
Its Jacobian at $x$ is $I-D^{-1}H(x)$, hence $I-H^{\ast}$ at the
solution; GFT-FP is thus a quasi-Newton method in which $H^{\ast}$ is
approximated by the identity. By
Theorem~\ref{prop:spec}, $I-H^{\ast}\succeq0$ with spectral radius
\begin{equation}\label{eq:rate}
\rho(I-H^{\ast})=1-\lambda_{\min}(H^{\ast})\leq1-\lambda_{\min}(C)<1 ,
\end{equation}
so the iteration is locally convergent for every $z$, a fact established
globally by \citet{ArchakovHansen:Correlation}. The spectral radius in
\eqref{eq:rate} is the sharp worst-case local factor: the asymptotic
factor along a trajectory can be smaller when the error avoids the slow
eigendirections, as under symmetric initialization for equicorrelation
matrices. The local linearized worst-case count to
tolerance $\epsilon$ has denominator
$-\log\{1-\lambda_{\min}(H^{\ast})\}$ and also depends on the
starting error; for small $\lambda_{\min}(H^{\ast})$ its leading
scaling is $|\log\epsilon|/\lambda_{\min}(H^{\ast})$. This scaling is
realized in practice: Section~\ref{sec:numerics} reports a design class
whose median worst-case factor corresponds to
$\lambda_{\min}(H^{\ast})=0{\cdot}028$, with a median count of
$1{,}008$ iterations, consistent with
$|\log10^{-13}|/0{\cdot}028\approx1{,}100$; Figure~\ref{fig:conv}(c)
illustrates the predicted scaling in the ill-conditioned regime. The bound \eqref{eq:rate} does not itself imply slow convergence, but
small $\lambda_{\min}(H^{\ast})$ is commonly observed for
near-singular $C$, as documented by \citet{ArchakovHansen:Correlation}
and \citet{ChenFeiYu:2025}.

Second, part (iii) holds at every $x$, not only at the solution, so
the Newton system $H(x)\delta=-\nabla f(x)$ is never worse conditioned
than $e^{A[x]}$, before or after diagonal scaling; this justifies
conjugate gradients with the diagonal preconditioner $D$ at every
iterate.

\section{The algorithm}\label{sec:algorithm}

All methods for solving \eqref{eq:varchar} share one $O(n^{3})$
kernel, the eigendecomposition
$A[x]=Q\Lambda Q'$, which delivers $e^{A[x]}$, $f$, $\nabla f$ and the
action of $H$; we therefore report costs in eigendecompositions and
Hessian-vector products. Let $L$ denote the matrix of divided
differences of the exponential,
$L_{kl}=(e^{\lambda_k}-e^{\lambda_l})/(\lambda_k-\lambda_l)$ when
$\lambda_k\neq\lambda_l$, with the continuous value
$L_{kl}=e^{\lambda_k}$ whenever $\lambda_k=\lambda_l$, including
repeated eigenvalues with $k\neq l$. The Daleckii--Krein formula gives,
for $v\in\R^{n}$,
\begin{equation}\label{eq:hv}
H(x)v=\diag\big[Q\{L\circ(Q'\diag(v)Q)\}Q'\big],
\end{equation}
two matrix multiplications per application; unqualified norms below are
Euclidean. Forming the full Hessian costs $O(n^{4})$, the Jacobian cost
paid at every iteration by Newton's method and once by Broyden's method
in \citet{ChenFeiYu:2025}; the algorithm below never forms it.

Two further consequences of Theorem~\ref{prop:spec}(i) are used
below. Since $A[x+c\one]=A[x]+cI$ for $c\in\R$, the eigenvectors of
$A[x]$ do not depend on $c$, its eigenvalues increase by $c$, and,
writing $\ell(x)=\log\diag(e^{A[x]})$,
\begin{equation}\label{eq:shift}
\ell(x+c\one)=\ell(x)+c\one,\qquad
D(x+c\one)^{-1}H(x+c\one)=D(x)^{-1}H(x),
\end{equation}
where $D^{-1}H$ is the Jacobian of $\ell$, which maps $\one$ to
$\one$ by Theorem~\ref{prop:spec}(i). The residual $\ell$, and with
it the fixed-point step \eqref{eq:fp}, is therefore exactly affine
along $\one$. Moreover, $f(x+c\one)=e^{c}\tr e^{A[x]}-\one'x-nc$ is
minimized exactly at $c=-s(x)$, with $s(x)=\log\{\tr(e^{A[x]})/n\}$,
and the normalized point $N(x)=x-s(x)\one$ satisfies
\begin{equation}\label{eq:norm}
\tr e^{A[N(x)]}=n,\qquad
f(x)-f\{N(x)\}=n\{e^{s(x)}-1-s(x)\}\geq0.
\end{equation}
Normalization needs no new eigendecomposition: $s(x)$ is a
log-sum-exp of the eigenvalues, the eigenvectors are unchanged, the
eigenvalues and $\ell$ decrease by $s(x)$, and $D$, $H$, and $L$ are
multiplied by $e^{-s(x)}$. At a normalized point
$\lambda_{\max}(A[x])\leq\log n$ and $f=n-\one'x$, so the objective
is evaluated without overflow wherever the iteration goes.

The algorithm is arranged so that every quantity is computed accurately
in floating point. It works with $\ell(x)$, obtained directly from
$(\Lambda,Q)$ without forming an exponential that could overflow, and
every evaluated point is normalized by \eqref{eq:norm}, so that
$\tr e^{A[x]}=n$ and the trace term never overflows; fixed-point
steps are taken only while some $\ell_i<-700$ (Section~\ref{sec:s2}).
Two safeguards, stated in the algorithm, handle the neighbourhood of
the solution: the full Newton step is taken untested when the
predicted decrease falls below what can be certified numerically, and
fixed-point steps complete the computation if progress stalls at the
rounding floor of $\diag(e^{A[x]})$, where they remain effective. The
cancellation-free evaluation of $\ell$, $g$, $L$, and $f$, and the
rationale for the two safeguards, are detailed in
Section~\ref{sec:s2}.

Step~4 computes a Newton step for the equation $\ell(x)=0$ that the
fixed point solves, rather than for $\nabla f(x)=0$. The Jacobian of
$\ell$ is $D^{-1}H$, so the step solves $H\delta=-D\ell$, with the
same matrix, preconditioner, and cost per conjugate-gradient iteration
as $H\delta=-g$; the two right-hand sides agree to first order at
$x^{\ast}$, since $g=D\ell+O(\|\ell\|^{2})$. By \eqref{eq:shift} the
linear model $\ell+D^{-1}H\delta$ is exact along $\one$: from
$x^{\ast}+c\one$ the full step returns to $x^{\ast}$ for every $c$,
whereas the step solving $H\delta=-g$ moves by $-(1-e^{-c})\one$. With
preconditioner $D$, the preconditioned residual $D^{-1}r$ is the
residual of the linearized equation, so the forcing test measures the
inner solve in the units of the outer residual. With forcing tolerance
$\eta_k=\min(1/2,\|\ell_k\|)$ the Newton phase converges locally
quadratically when the forcing test is met within the product cap and
full steps are accepted \citep{DemboEisenstatSteihaug:1982}, as they
are near $x^{\ast}$; the choice $\eta_k=\min(1/2,\|\ell_k\|^{1/2})$
would give order $3/2$ \citep{EisenstatWalker:1996}. Unlike
$H\delta=-g$, the system $H\delta=-D\ell$ need not yield a descent
direction for $f$ far from $x^{\ast}$ (an example is given in
Section~\ref{sec:variants}). Step~5 therefore tests descent first and
otherwise takes a fixed-point step; this never occurred in the
experiments of Section~\ref{sec:numerics}. Each tested Newton step is
an Armijo descent step for $f$, and so is the fixed-point step.

The diagonal preconditioner can be replaced by a curvature model built
from the same eigendecomposition. Since
$H=\int_0^1e^{tA}\circ e^{(1-t)A}dt$ by \eqref{eq:quadform}, an
$r$-point Gauss--Legendre rule $(t_j,a_j)$ on $[0,1]$ gives
\begin{equation}\label{eq:Mr}
M_r=\sum_{j=1}^{r}a_j\,e^{t_jA}\circ e^{(1-t_j)A},\qquad A=A[x],
\end{equation}
which is positive definite by the Schur product theorem and costs $r$
matrix multiplications, those forming $e^{t_jA}=Q\diag(e^{t_j\lambda})Q'$,
and one Cholesky factorization. Because Gauss--Legendre quadrature
underestimates $\int_0^1e^{tu}dt$ for every real $u$,
\begin{equation}\label{eq:cert}
M_r\preceq H\preceq\phi_r(\Delta)\,M_r,\qquad
\Delta=\lambda_{\max}(A[x])-\lambda_{\min}(A[x]),
\end{equation}
where $\phi_r$ is an explicit scalar function of the spectral spread
alone (Section~\ref{sec:quad}); at the solution
$\Delta=\log\operatorname{cond}(C)$. Two nodes certify a condition
number of at most two for the preconditioned system whenever
$\operatorname{cond}(C)\leq1{\cdot}1\times10^{5}$, and three nodes
whenever $\operatorname{cond}(C)\leq5\times10^{9}$, so the order
follows from the extreme eigenvalues at no extra cost. Conjugate
gradients preconditioned by the Cholesky factor of $M_r$ then need two
to four products per Newton step, and the forcing test is still
applied to $D^{-1}r$, so the outer iteration is unchanged. The
construction pays whenever two or more nodes are needed, that is
whenever the diagonal preconditioner is slow, and increasingly with
$n$ as products become expensive relative to eigendecompositions:
one-factor designs are inverted $1{\cdot}2$ times faster at $n=100$
and $1{\cdot}5$ times faster at $n=800$. Section~\ref{sec:quad} gives
the selection rule and the evidence; the diagonal preconditioner
remains the default.

\begin{proposition}\label{prop:descent}
For every $x\in\R^{n}$, with $\ell=\ell(x)$,
\[
f(x-\ell)\leq f(x)-\sum_{i=1}^{n}(e^{\ell_{i}}-1-\ell_{i}),
\]
and the subtracted amount is strictly positive unless $x=x^{\ast}(z)$.
\end{proposition}

\begin{proof}
The Golden--Thompson inequality $\tr(e^{B+C})\leq\tr(e^{B}e^{C})$ for
symmetric $B,C$, with $C=-\diag(\ell)$ diagonal, gives
$\tr(e^{A[x-\ell]})\leq\sum_{i}(e^{A[x]})_{ii}e^{-\ell_{i}}=n$.
Subtract $\one'(x-\ell)$ and compare with
$f(x)=\sum_{i}e^{\ell_{i}}-\one'x$; positivity is $e^{t}-1-t>0$ for
$t\neq0$, and $\ell=0$ characterizes $x^{\ast}(z)$.
\end{proof}

Telescoping the bound gives
$\sum_{k}\sum_{i}(e^{\ell_{ki}}-1-\ell_{ki})\leq f(x_{0})-\min f
<\infty$, so $\ell_{k}\to0$; coercivity makes the iterates
precompact, and uniqueness of the stationary point identifies every
cluster point with $x^{\ast}(z)$: a short new proof of the global
convergence of GFT-FP, established with a contraction argument by
\citet{ArchakovHansen:Correlation}. Thus every fixed-point step and
every tested Newton step decreases the same objective; only the
untested Newton steps of step~5 are not guaranteed to decrease $f$.

The same argument covers any acceleration safeguarded by the decrease
in Proposition~\ref{prop:descent}. Let
$V(x)=\sum_i(e^{\ell_i}-1-\ell_i)$ and fix $\theta\in(0,1)$. Suppose
a trial point $y$, produced by Anderson mixing, a Newton step, or any
other rule, is accepted only when $f(y)\leq f(x)-\theta V(x)$, and the
fixed-point step is taken otherwise. Then every step decreases $f$ by
at least $\theta V(x)$, and the iterates converge to $x^{\ast}(z)$
from every starting point (Section~\ref{sec:variants}). The safeguard does
not slow local convergence: near $x^{\ast}$ the full Newton step of
step~4 decreases $f$ by
$\tfrac12\ell'H^{\ast-1}\ell+O(\|\ell\|^{3})\geq\tfrac12\|\ell\|^{2}+O(\|\ell\|^{3})$,
because $\lambda_{\max}(H^{\ast})=1$, while
$V(x)=\tfrac12\|\ell\|^{2}+O(\|\ell\|^{3})$.
Section~\ref{sec:tools} uses the safeguard for Anderson acceleration.

\begin{center}
\fbox{\parbox{0.93\textwidth}{
\textbf{The GFT-FP+N algorithm.} Input $z\in\R^{d}$; set $x=0$,
$t_{\mathrm{prev}}=1$, and $\texttt{tol}=10^{-13}$.
\begin{enumerate}\itemsep1pt
\item Eigendecompose $A[x]=Q\Lambda Q'$; compute
$\ell=\log\diag(e^{A[x]})$ from $(\Lambda,Q)$. Normalize: with
$s=\log(\sum_k e^{\lambda_k}/n)$, set $x\leftarrow x-s\one$,
$\Lambda\leftarrow\Lambda-sI$, and $\ell\leftarrow\ell-s\one$.
\item If $\min_i\ell_i<-700$: fixed-point step $x\leftarrow x-\ell$;
set $t_{\mathrm{prev}}=1$; go to 1.
\item Set $g=e^{\ell}-\one$, evaluated without cancellation; stop if
$\|g\|_{\infty}<\texttt{tol}$, returning $C=e^{A[x]}$. If
$\|g\|_{\infty}<10^{-9}$ has failed to halve over three consecutive
iterations, enter terminal mode; in terminal mode take the fixed-point
step $x\leftarrow x-\ell$ and go to 1.
\item Solve $H\delta=-D\ell$, $D=\diag(e^{\ell})$, by conjugate
gradients started at $\delta=0$, with residual $r=H\delta+D\ell$,
diagonal preconditioner $D$ (or, when selected by the rule of
Section~\ref{sec:quad}, the quadrature preconditioner \eqref{eq:Mr},
with the same stopping test) and products \eqref{eq:hv}, stopping when
$\|D^{-1}r\|\leq\eta\|\ell\|$ with $\eta=\min(1/2,\|\ell\|)$,
capped at $2n$ products.
\item If $\delta$ is not finite or $g'\delta\geq0$: fixed-point step,
$t_{\mathrm{prev}}=1$; go to 1. If $|g'\delta|\leq10^{-12}(1+|f|)$
and $\|\ell\|_{\infty}<10^{-3}$, with $f=\sum_k e^{\lambda_k}-\sum_i x_i$:
set $x\leftarrow x+\delta$; go to 1. Otherwise backtrack on $f$ along
$\delta$ by Armijo halving from $t=\min(1,2t_{\mathrm{prev}})$,
testing the normalized trial points $N(x+t\delta)$, and set
$t_{\mathrm{prev}}=t$ on success; if the search fails, substitute a
fixed-point step and set $t_{\mathrm{prev}}=1$. Go to 1.
\end{enumerate}}}
\end{center}

\section{Numerical results}\label{sec:numerics}

We compare four methods: the fixed point \eqref{eq:fp}; Broyden's
method, implemented as published by \citet{ChenFeiYu:2025}, with the
exact Jacobian computed once after one fixed-point step and rank-one
updates thereafter; full Newton, by which we mean GFT-FP+N itself
with the Newton system $H\delta=-D\ell$ solved exactly, the Hessian
formed explicitly at $O(n^{4})$ cost at every iteration, in place of
the matrix-free conjugate gradients, so that the two differ only in
the linear solver; and GFT-FP+N. The version of GFT-FP+N with the
gradient system $H\delta=-g$, square-root forcing, and a fixed-point
phase while $\max_i\ell_i>\log2$ is compared with the algorithm of
Section~\ref{sec:algorithm} in Section~\ref{sec:variants}, which
varies these choices one at a time under otherwise identical rules.
Newton's method in the form of
\citet{ChenFeiYu:2025}, one fixed-point step followed by exact Newton
steps, to which we add a line search, is examined in
Section~\ref{sec:newtonsafe}, and two standard
Jacobian-free tools, Anderson acceleration of the fixed point and
limited-memory BFGS on $f$, in Section~\ref{sec:tools}. All methods run to
$\|\diag(e^{A[x]})-\one\|_{\infty}\leq10^{-13}$ with iteration caps of
$5{,}000$ for the fixed point and for Anderson acceleration and $500$
otherwise. For designs
generated from a known matrix, reconstruction error against that
matrix is of the order of the tolerance; for $z$-generated designs the
stopping criterion itself is the error measure. The designs differ
enormously in conditioning: the median base-ten logarithm of the
condition number of the solution $C$ is $1{\cdot}0$, $2{\cdot}5$ and
$4{\cdot}2$ for the three Toeplitz designs and $1{\cdot}5$ and
$4{\cdot}3$ for the Wishart and one-factor designs, but $23$ and $46$
for $\sigma=2$ and $\sigma=4$. The $z$-generated designs therefore test
the solution of the diagonal-normalization equations to residual
$10^{-13}$, not the accuracy of the returned dense matrix: for
$\sigma\geq2$ the returned matrix, reconstructed as $Qe^{\Lambda}Q'$
with unit diagonal, is no longer numerically positive definite, its
computed spectrum containing negative values of order $10^{-15}$, so
it supports neither Cholesky factorization nor evaluation of its
logarithm. Likelihood applications operate in the
well-conditioned regime of the first five designs, where the returned
matrix is accurate to the tolerance. Test
problems comprise the Toeplitz design $C_{ij}=\rho^{|i-j|}$ of
\citet{ChenFeiYu:2025} with $\rho\in\{0{\cdot}5,0{\cdot}9,0{\cdot}99\}$;
Wishart-type matrices $C=\Delta S\Delta$, $S=XX'$ with $X$ an
$n\times2n$ standard Gaussian matrix and $\Delta$ the diagonal
normalization; one-factor matrices
$C=\beta\beta'+\diag(1-\beta^{2})$ with
$\beta_i\sim\mathrm{Un}(0{\cdot}85,0{\cdot}999)$ at $n=100$ and
$\mathrm{Un}(0{\cdot}8,0{\cdot}995)$ at $n=800$; and unstructured
vectors $z$ with independent $N(0,\sigma^{2})$ elements,
$\sigma\in\{2,4\}$. Random designs use $1{,}000$ independent draws,
$500$ at $n=800$; we report medians over converging replications of the count of
eigendecompositions, of the number of Hessian-vector products for
GFT-FP+N, and of the wall time, together with the number of failures
to converge; on the deterministic Toeplitz designs the wall time is
the minimum of five repeated timings. In Table~\ref{tab:bench} a dash marks a method
not run at that dimension, owing to its Jacobian cost. Computations use double
precision in Julia 1.13.0 with single-threaded OpenBLAS on an Apple
M2 Max processor; because wall times depend on language and BLAS build,
the counts are the primary comparison, and they were reproduced exactly
on the deterministic designs, and closely on the random designs, by an
independent NumPy implementation on serialized inputs; seeds and
replication code are described in Section~\ref{sec:s1}.

\begin{table}
\centering
\caption{For each method, the median number of eigendecompositions and
the median wall time per inversion (in parentheses; minimum of five
timings on the Toeplitz designs) to reach tolerance
$10^{-13}$, with failures as superscripts (out of 1{,}000 draws; 500 at
$n=800$); for GFT-FP+N also the median number of Hessian--vector
products (in braces); the final column is the median worst-case local
factor of the fixed point.}\label{tab:bench}
\medskip
\footnotesize
\setlength{\tabcolsep}{3pt}
\begin{tabular}{lccccc}
\toprule
Case & GFT-FP & Broyden & Full Newton & GFT-FP+N & $1-\lambda_{\min}(H^{\ast})$\\
\midrule
Toeplitz $\rho=0{\cdot}5$, $n=100$  & 15 (11ms) & 5 (10ms) & 4 (20ms) & 5 \{15\} (6ms) & $0{\cdot}16$\\
Toeplitz $\rho=0{\cdot}9$, $n=100$  & 41 (31ms) & 8 (12ms) & 5 (26ms) & 6 \{30\} (8ms) & $0{\cdot}55$\\
Toeplitz $\rho=0{\cdot}99$, $n=100$ & 89 (68ms) & 8 (12ms) & 5 (26ms) & 6 \{29\} (8ms) & $0{\cdot}77$\\
Toeplitz $\rho=0{\cdot}99$, $n=300$ & 89 (641ms) & 8 (387ms) & -- & 6 \{28\} (108ms) & $0{\cdot}77$\\
Wishart, $n=100$                    & 15 (12ms) & 6 (10ms) & 4 (18ms) & 5 \{15\} (6ms) & $0{\cdot}13$\\
One-factor, $n=100$                 & 124 (130ms) & 16 (22ms) & 6 (32ms) & 7 \{34\} (11ms) & $0{\cdot}79$\\
$z\sim N(0,4I_d)$, $n=50$           & 473 (92ms) & 38 (8ms)$^{56}$ & 9 (5ms) & 10 \{68\} (3ms) & $0{\cdot}94$\\
$z\sim N(0,16I_d)$, $n=50$          & 1{,}008 (197ms) & 52 (10ms)$^{978}$ & 12 (6ms) & 14 \{93\} (4ms) & $0{\cdot}97$\\
One-factor, $n=800$                 & 136 ($15{\cdot}5$s) & -- & -- & 8 \{35\} ($2{\cdot}3$s) & --\\
\bottomrule
\end{tabular}
\end{table}

We observe three patterns in Table~\ref{tab:bench}. First, on well-conditioned
designs all four methods are inexpensive. Second, as $1-\lambda_{\min}(H^{\ast})$ approaches one the
fixed-point count grows in line with \eqref{eq:rate}, reaching a median
of $1{,}008$ iterations, whereas the counts for the three Newton-type
methods change little. Among these, Broyden's method failed in $56$ of
$1{,}000$ replications at $\sigma=2$ and in $978$ of $1{,}000$ at
$\sigma=4$, apparently because its one-time Jacobian is formed too far
from the solution: given an initial fixed-point phase in the log
domain, its failures drop from $978$ to $167$ of $1{,}000$ at
$\sigma=4$, at three times the cost of GFT-FP+N and still with the
$O(n^{4})$ Jacobian (Section~\ref{sec:s5}); at the tolerance
$10^{-6}$ of \citet{ChenFeiYu:2025}, on the same inputs, its failures
are unchanged (Section~\ref{sec:s4}), so they reflect divergence
rather than the tolerance. Full Newton and GFT-FP+N converged in
every replication of every design, including the $500$ at $n=800$,
where the tolerance $10^{-13}$ sits at the rounding floor
(Section~\ref{sec:s6}). Full Newton takes one or two fewer
eigendecompositions because its Newton directions are exact rather
than truncated, but at a cost per iteration that grows with the
explicit Hessian: it takes $1{\cdot}3$ to $1{\cdot}4$ times longer at
$n=50$ and $2{\cdot}9$ to $3{\cdot}5$ times longer at $n=100$, and it
is not run at $n\geq300$. Quadrature preconditioning \eqref{eq:Mr}
recovers most of this reduction without forming $H$: with two nodes,
the one-factor design at $n=800$ needs $6$ eigendecompositions and
$16$ products in place of $8$ and $35$, and $32$\% less time; on
the designs of Table~\ref{tab:bench} with $n\leq300$ the selection
rule of Section~\ref{sec:quad} leaves the diagonal preconditioner in
place on the two designs where one node certifies the bound and
selects $M_r$ on the others, with gains of up to $1{\cdot}4$ at
$n=100$ (Section~\ref{sec:quad}). Computed from unrounded median times,
GFT-FP+N's wall time was lower than the fixed point's by a factor of
up to $45$ ($\sigma=4$ design) and $6{\cdot}6$ ($n=800$), and lower
than Broyden's by a factor of $3{\cdot}6$ at $n=300$ and $1{\cdot}9$
on the one-factor design at $n=100$. Third,
the local factors in the last column track the observed fixed-point
counts; Figure~\ref{fig:conv}(c), with $1{,}000$ random designs,
illustrates the predicted scaling in the ill-conditioned regime.

Warm starts matter in dynamic applications. In a random-walk sequence of $50$
vectors ($n=100$, one-factor initial matrix, increments of standard
deviation $0{\cdot}002$ per coordinate, aggregate step about
$0{\cdot}14$), sequential inversion with each solution carried
forward and Broyden's Jacobian recomputed at each step required, on
average over ten independent sequences and including the initial
cold-start inversion,
$4{\cdot}8$ eigendecompositions per step for GFT-FP+N, $6{\cdot}2$
for Broyden's method, and $103$ for the fixed point, at mean wall
times per step of $8{\cdot}8$, $11{\cdot}8$, and $105$ milliseconds;
Broyden's Jacobian, formed at every step, accounts for its higher time
at nearly equal counts. In the terms of the applications cited in the
introduction, the inversion component of one likelihood evaluation
over this path of fifty time points, each requiring one inversion,
costs approximately $0{\cdot}44$ seconds with GFT-FP+N, $0{\cdot}59$
with Broyden's method, and $5{\cdot}2$ with the fixed point, the
log-determinant terms being free by Section~\ref{sec:discussion}.
Warm starts thus bring Broyden's method close to GFT-FP+N but do not
rescue the fixed point on ill-conditioned designs.

A first-order predictor lowers these counts further. When $z_t$
arrives, the eigendecomposition at the previous solution $x_{t-1}$ is
still available, and differentiating $\ell\{x^{\ast}(z);z\}=0$ as in
Section~\ref{sec:discussion} gives the prediction
\begin{equation}\label{eq:pred}
\hat{x}_t=x_{t-1}+p,\qquad
H(x_{t-1})\,p=-\{B\Delta z_t+D\ell(x_{t-1})\}.
\end{equation}
Here $\Delta z_t=z_t-z_{t-1}$, and
$B\Delta z_t=\diag[Q\{L\circ(Q'\Delta A_tQ)\}Q']$ is the derivative of
$\diag(e^{A})$ in the direction of the off-diagonal change
$\Delta A_t$, computed with the stored $Q$ and $L$ at the cost of three
matrix multiplications; the term $D\ell$ corrects for the residual
left by the previous solve. The system is solved by conjugate
gradients with preconditioner $D$ and products \eqref{eq:hv} to
relative residual $10^{-3}$, and no eigendecomposition is needed. The
prediction error is $O(h^{2}+\tau+\eta_ph)$ for step size
$h=\|\Delta z_t\|$, previous-solve error $\tau$, and relative solve
tolerance $\eta_p$. In the design above the predictor reduced the mean
count to $3{\cdot}1$ eigendecompositions per step, three per warm
step, and the mean time per step from $8{\cdot}8$ to $6{\cdot}4$
milliseconds, so the inversion component of the likelihood path costs
$0{\cdot}32$ seconds; the gain depends on the step size
(Section~\ref{sec:variants}).

\begin{figure}
\centering
\includegraphics[width=0.98\textwidth]{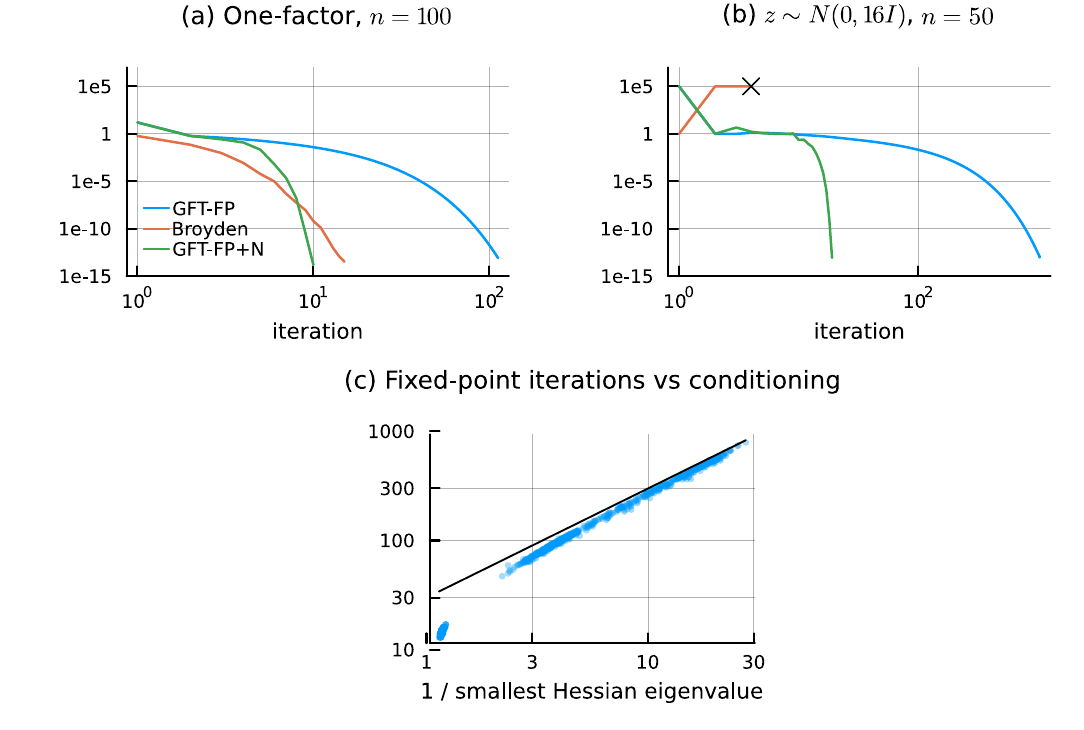}
\caption{Convergence histories, error
$\|\diag(e^{A[x]})-\one\|_{\infty}$ per iteration, logarithmic axes,
for (a) a one-factor design, $n=100$, and (b) $z\sim N(0,16I_d)$, $n=50$,
with a cross marking non-convergence of Broyden's method; and (c)
fixed-point iteration counts against $1/\lambda_{\min}(H^{\ast})$ for
1{,}000 random designs, with $|\log\epsilon|/\lambda_{\min}(H^{\ast})$,
$\epsilon=10^{-13}$, as the solid line.}
\label{fig:conv}
\end{figure}

\subsection{Comparisons on identical inputs}\label{sec:identical}

The four comparisons below use the serialized draws behind
Table~\ref{tab:bench}, so that methods are compared on identical
inputs; each entry reports the median count of eigendecompositions
over converging replications, the median wall time in milliseconds in
parentheses, and the number of failures out of $1{,}000$ as a
superscript.

\subsubsection{Tolerance sensitivity}\label{sec:s4}

Table~\ref{tab:tol} reports the $z\sim N(0,16I_d)$, $n=50$ design with
the same $1{,}000$ serialized input vectors used in
Table~\ref{tab:bench},
solved to tolerances $10^{-6}$, as in \citet{ChenFeiYu:2025}, and
$10^{-13}$, for the fixed point, Broyden's method, Newton with Armijo backtracking
(Section~\ref{sec:newtonsafe}), and GFT-FP+N. Broyden's
method fails in the same $978$ replications at both tolerances, so
the tolerance is not the cause; inspection of the failed runs shows
divergence, with non-finite iterates or residuals growing without
bound. The $94$ failures of Newton with Armijo backtracking at
$10^{-13}$ vanish at
$10^{-6}$, consistent with line-search stagnation at the
floating-point resolution of the objective.

\begin{table}[h]
\centering
\caption{Median eigendecompositions over converging replications, median
wall time in milliseconds (in parentheses) and failures out of
1{,}000 (superscripts), $z\sim N(0,16I_d)$, $n=50$, same inputs at
both tolerances.}\label{tab:tol}
\medskip
\begin{tabular}{lcccc}
\toprule
 & GFT-FP & Broyden & Newton, Armijo & GFT-FP+N\\
\midrule
Tolerance $10^{-6}$  & $440{\cdot}5$ (79) & $35{\cdot}0$ (7)$^{978}$ & $21{\cdot}0$ (7) & $13{\cdot}0$ (3)\\
Tolerance $10^{-13}$ & $1007{\cdot}5$ (181) & $52{\cdot}0$ (10)$^{978}$ & $24{\cdot}0$ (8)$^{94}$ & $14{\cdot}0$ (4)\\
\bottomrule
\end{tabular}
\end{table}

\subsubsection{Broyden's method with an initial fixed-point phase}\label{sec:s5}

The published algorithm of \citet{ChenFeiYu:2025} precedes Broyden
updates with a single fixed-point step. To separate the effect of this
initialization from the intrinsic behaviour of Broyden updates,
Table~\ref{tab:glob} compares the published method with a variant given
an initial fixed-point phase in the log domain, fixed-point steps
until $\max_i\ell_i\leq\log2$ as in the first-draft version of
GFT-FP+N (Section~\ref{sec:variants}), after which the Jacobian is
formed once and rank-one updates proceed as published, without
GFT-FP+N's line search or terminal safeguard
(\texttt{inv\_gft\_broyden} with \texttt{globalized = true}). The
inputs are the same $1{,}000$ serialized draws per design as in
Table~\ref{tab:bench}.

\begin{table}[h]
\centering
\caption{Median eigendecompositions over converging replications, median
wall time in milliseconds (in parentheses) and failures out of
1{,}000 (superscripts), extreme designs with $n=50$, tolerance
$10^{-13}$, same inputs across rows.}\label{tab:glob}
\medskip
\begin{tabular}{lccc}
\toprule
 & Broyden, published & Broyden, globalized & GFT-FP+N\\
\midrule
$z\sim N(0,4I_d)$  & $38{\cdot}0$ (7)$^{56}$  & $29{\cdot}0$ (6)$^{3}$   & $10{\cdot}0$ (3)\\
$z\sim N(0,16I_d)$ & $52{\cdot}0$ (10)$^{978}$ & $43{\cdot}0$ (9)$^{167}$ & $14{\cdot}0$ (4)\\
\bottomrule
\end{tabular}
\end{table}

The shared initial phase removes most of the failures of the published
method, from $978$ to $167$ of $1{,}000$ on the harder design, confirming
that the one-step initialization forms the Jacobian too far from the
solution. A substantial residual failure rate remains, however, and
the variant still requires about three times the
eigendecompositions of GFT-FP+N, in addition to its one-time
$O(n^{4})$ Jacobian, which also caps the feasible dimension. The
comparison thus attributes most, but not all, of Broyden's failures
to its initialization, without altering the ordering of the methods.

\subsubsection{Newton with Armijo backtracking}\label{sec:newtonsafe}

The full-Newton column of Table~\ref{tab:bench} is GFT-FP+N with the
exact linear solver. A second Newton comparator follows the form of
\citet{ChenFeiYu:2025}, one fixed-point step followed by exact Newton
steps, to which we add an Armijo line search on $f$, valid since
$H\succ0$ makes $\delta$ an $f$-descent direction; without it the
iteration diverged on most designs, and with it the comparator is our
adaptation rather than the published algorithm, so its behaviour
should not be read as that of the original. It solves the gradient
system $H\delta=-g$ and has none of the normalization, the
untested-step rule, or the terminal safeguard of GFT-FP+N
(\texttt{inv\_gft\_newton}). Table~\ref{tab:newtonsafe}
places the two side by side on every design with $n\leq100$, on the
same draws as Table~\ref{tab:bench}.

\begin{table}[h]
\centering
\caption{Median eigendecompositions over converging replications, median
wall time in milliseconds (in parentheses) and failures out of 1{,}000
(superscripts), tolerance $10^{-13}$, same inputs across
columns.}\label{tab:newtonsafe}
\medskip
\begin{tabular}{lccc}
\toprule
Case & Newton, Armijo & Full Newton & GFT-FP+N\\
\midrule
Toeplitz $\rho=0{\cdot}5$, $n=100$  & 5 (20) & 4 (20) & 5 (6)\\
Toeplitz $\rho=0{\cdot}9$, $n=100$  & 6 (26) & 5 (26) & 6 (8)\\
Toeplitz $\rho=0{\cdot}99$, $n=100$ & 7 (33) & 5 (26) & 6 (8)\\
Wishart, $n=100$                    & 5 (19)$^{140}$ & 4 (18) & 5 (6)\\
One-factor, $n=100$                 & 7 (34)$^{63}$ & 6 (32) & 7 (11)\\
$z\sim N(0,4I_d)$, $n=50$           & 14 (6)$^{82}$ & 9 (5) & 10 (3)\\
$z\sim N(0,16I_d)$, $n=50$          & 24 (9)$^{94}$ & 12 (6) & 14 (4)\\
\bottomrule
\end{tabular}
\end{table}

Newton with Armijo backtracking fails in up to $14$\% of
replications, even on the well-conditioned Wishart design, through
stagnation of the line search at the floating-point resolution of
$f$, the regime step~5 of GFT-FP+N is designed to traverse; its counts
on the deterministic designs are also unstable across linear-algebra
builds (Section~\ref{sec:s3}). With the safeguards of GFT-FP+N, exact
Newton steps converge in every replication at the same or lower cost.
The failures are therefore attributable to the missing safeguards,
not to Newton's method as such, and the comparison that matters,
between full Newton and GFT-FP+N, reduces to cost per iteration:
explicit $O(n^{4})$ Hessian formation against matrix-free products at
two matrix multiplications each.

\subsubsection{Anderson acceleration and limited-memory BFGS}\label{sec:tools}

Two general-purpose tools apply directly to the problem without a
Jacobian. Anderson acceleration \citep{WalkerNi:2011} of the
fixed-point map $x\mapsto x-\ell(x)$, with the residual $-\ell(x)$
evaluated in the log domain and memory $5$, costs one
eigendecomposition per iteration plus a small least-squares solve
(\texttt{inv\_gft\_anderson}). Limited-memory BFGS
\citep{LiuNocedal:1989} on $f$ with memory $10$, two-loop recursion,
the usual initial scaling and Armijo backtracking on $f$ costs one
eigendecomposition per function evaluation; it is run from $x=0$ and,
as for Broyden's method in Section~\ref{sec:s5}, after an initial
log-domain fixed-point phase until $\max_i\ell_i\leq\log2$
(\texttt{inv\_gft\_lbfgs}, option \texttt{globalized}); both
L-BFGS variants are given the untested-step rule of GFT-FP+N near the
rounding floor, without which their Armijo backtracking stalls there,
while Anderson acceleration has no line search and needs no such rule.
A guarded variant of Anderson acceleration accepts a proposal only
when it passes the test of Section~\ref{sec:algorithm} with
$\theta=1/4$, untested below the certification threshold of step~5,
and otherwise clears its memory and takes the fixed-point step
(\texttt{inv\_gft\_anderson}, option \texttt{guarded}).
Table~\ref{tab:tools}
reports the two extreme designs on the draws of Table~\ref{tab:bench}.

\begin{table}[h]
\centering
\caption{Median eigendecompositions over converging replications, median
wall time in milliseconds (in parentheses) and failures out of
1{,}000 (superscripts), extreme designs with $n=50$, tolerance
$10^{-13}$, same inputs across columns.}\label{tab:tools}
\medskip
\small
\begin{tabular}{lccccc}
\toprule
 & Anderson & Anderson, guarded & L-BFGS & L-BFGS, globalized & GFT-FP+N\\
\midrule
$z\sim N(0,4I_d)$  & $92{\cdot}0$ (17) & $92{\cdot}0$ (17) & $88{\cdot}0$ (16) & $59{\cdot}0$ (11) & $10{\cdot}0$ (3)\\
$z\sim N(0,16I_d)$ & $157{\cdot}0$ (29) & $158{\cdot}0$ (30) & $153{\cdot}0$ (28) & $88{\cdot}0$ (16) & $14{\cdot}0$ (4)\\
\bottomrule
\end{tabular}
\end{table}

Both tools converge in every replication on these inputs. The
ordinary Anderson scheme and this L-BFGS variant with untested
terminal steps have no convergence guarantee here, however, and with
larger memory Anderson acceleration does fail: with memory $50$ it
diverged on $9$ of the $1{,}000$ draws of the $z\sim N(0,16I_d)$
design (Section~\ref{sec:variants}). The guarded variant converges
from every starting point in exact arithmetic, by the argument of
Section~\ref{sec:algorithm}; on these inputs it needs essentially the
same number of eigendecompositions as the ordinary scheme, at one
percent more time for the evaluations of $f$, and at memory $50$ it
removes the divergences. Their cost is the point. Anderson acceleration and L-BFGS from $x=0$
need nine to eleven times the eigendecompositions of GFT-FP+N, and
the globalized L-BFGS still needs about six times as many, at
correspondingly higher wall times.
The gap is the difference between secant information accumulated over
a memory of ten steps and the exact curvature that the Hessian--vector
products supply to conjugate gradients at a cost of two matrix
multiplications each.

\section{Discussion}\label{sec:discussion}

Recovering $C$ from $\gamma$ is an exponential-family analogue of
matrix balancing: the fixed point plays the role of Sinkhorn scaling
\citep{Sinkhorn:1964}, GFT-FP+N that of its Newton replacements
\citep{KnightRuiz:2013}. Two extensions appear promising. First, GFT-FP+N uses only
exponential actions, Fr\'echet-derivative actions, available via an
augmented matrix \citep{AlMohyHigham:2009}, and extraction of
$\diag(e^{A})$; for banded, sparse, or Toeplitz log-correlation
structures the first two reduce to matrix-vector products, while the
diagonal requires $n$ actions or a stochastic estimator, so the
balance of costs in a matrix-free variant for large structured
problems is a question in its own right, left for future work. Second, differentiating
$\diag(e^{A[x^{\ast}(z);z]})=\one$ gives
$\partial x^{\ast}/\partial z'=-H^{\ast-1}B$, with $B$ the Jacobian of
$\diag(e^{A})$ in the off-diagonal elements, so delta-method standard
errors and differentiation through $\gamma^{-1}$ involve exactly the
matrix $H^{\ast}$ and the same conjugate-gradient machinery; applied
to a single direction $\Delta z$, the same formula gives the tangent
predictor of Section~\ref{sec:numerics}, which reuses the previous
eigendecomposition and costs only products. For large $n$, where a
Hessian--vector product costs about a third of an eigendecomposition,
the same eigendecomposition also yields the certified quadrature
preconditioner \eqref{eq:Mr}, whose required order is set by the
condition number of $C$. Finally,
$\log\det\gamma^{-1}(z)=\tr A[x^{\ast}]=\one'x^{\ast}(z)$, so the
log-determinant term of likelihoods under this parametrization costs
no further linear algebra.

\section*{Supplementary materials}
\begin{description}
\item[Appendix] Code and reproduction, numerically safe evaluation of
the algorithm, cross-language agreement on serialized inputs, and the
$n=800$ replications, the variants, safeguarded acceleration and
sequential-inversion experiments, and quadrature preconditioning
(Sections~S1 to S6 of this manuscript).
\item[Replication package] The Julia package \texttt{GFT.jl} with
all scripts, seeds, and logged outputs behind the results:
\texttt{https://github.com/reinhardhansen/GFT}.
\item[R package] \texttt{GFT}, available from CRAN.
\end{description}

\section*{Declaration of the use of generative AI and AI-assisted
technologies}
During the preparation of this work the authors used large language
models (Claude, Anthropic; ChatGPT, OpenAI) to assist with
mathematical exploration and checking, code development and
verification, numerical experiments, and drafting; the
transformation, the fixed-point method and the Fr\'echet-derivative
analysis underlying the Hessian representation originate in the
authors' earlier work and unpublished notes. The authors reviewed and
independently verified all output, edited it as necessary, and take
full responsibility for the content of the publication.

\section*{Acknowledgement}
The first author gratefully acknowledges financial support from the
Social Sciences and Humanities Research Council of Canada (SSHRC)
through an Insight Grant [435-2026-1701].

\appendix
\setcounter{section}{0}
\setcounter{table}{0}
\renewcommand{\thesection}{S\arabic{section}}
\renewcommand{\thetable}{S\arabic{table}}

\begin{center}
{\large\bf Appendix: supplementary material}
\end{center}

\section{Code and reproduction}\label{sec:s1}

The replication package contains the reference implementation and all
scripts, seeds, and outputs behind the results in the paper, and is
available at \texttt{https://github.com/reinhardhansen/GFT}. Its
structure is as follows. \texttt{GFT.jl} is a standard Julia package
(installable from the repository, after which \texttt{using GFT}
suffices; standard libraries only, Julia $\geq1.6$):
\texttt{gft} computes the forward
transformation, \texttt{inv\_gft} implements GFT-FP+N, and
\texttt{inv\_gft\_fp}, \texttt{inv\_gft\_broyden} (with an option
\texttt{globalized} described in Section~\ref{sec:s5}) and
\texttt{inv\_gft\_newton} implement the comparison methods, the
last with an argument \texttt{safeguard} whose value \texttt{false}
gives Newton with Armijo backtracking (Section~\ref{sec:newtonsafe}) and
whose default \texttt{true} adds the terminal safeguards of GFT-FP+N.
The full-Newton column of Table~\ref{tab:bench} is \texttt{inv\_gft}
with \texttt{exact\_hess = true}, the same algorithm as GFT-FP+N with
the explicit Hessian in place of conjugate gradients;
\texttt{inv\_gft\_anderson} (with an option \texttt{guarded}
implementing the safeguard of Section~\ref{sec:algorithm}) and
\texttt{inv\_gft\_lbfgs} implement the tools of
Section~\ref{sec:tools}; \texttt{inv\_gft\_path} inverts a sequence
of vectors with warm starts and, with the option \texttt{predictor},
the tangent predictor \eqref{eq:pred}, which \texttt{gft\_predict}
computes. The variants of Section~\ref{sec:variants} are options of
\texttt{inv\_gft}: the version of GFT-FP+N in the first draft of this
paper is \texttt{inv\_gft(z; residual = :gradient, forcing = :sqrt,
normalize = false, phase = true, adaptive = false)}; the option
\texttt{preconditioner} selects the diagonal preconditioner
(\texttt{:diagonal}, the default), the quadrature preconditioner at
every Newton step (\texttt{:quadrature}), or the rule of
Section~\ref{sec:quad} (\texttt{:auto}), and \texttt{precond.jl}
runs the experiments of that section. The R package
\texttt{GFT} on CRAN (version 1.2.0; base R, no dependencies) is a
function-for-function port of the Julia package under the same
conventions and argument names (with \texttt{"log"} for
\texttt{:log}, and so on), including the variants, the preconditioner
options, the tangent predictor, and the comparison solvers, with its
own golden-value tests. The benchmarks of the paper use the Julia
implementation at release v1.2.0.
\texttt{runtests.jl} is the test suite, whose golden values were
generated by an independent NumPy implementation
(Section~\ref{sec:s3}); \texttt{overnight.jl} runs the full benchmark
protocol and
writes the comma-separated files in \texttt{results/} from which
Table~\ref{tab:bench} and Figure~\ref{fig:conv} are generated
mechanically by the
notebook \texttt{figures.ipynb}; \texttt{final\_checks.jl} runs the
same-input comparisons of
Sections~\ref{sec:s4} to \ref{sec:tools}; \texttt{ablation.jl} runs
the experiments of Section~\ref{sec:variants};
\texttt{check\_serialized.jl}, with the inputs and NumPy counts in
\texttt{serialized/}, runs the cross-language agreement check of
Section~\ref{sec:s3}; \texttt{recheck\_n800.jl}
isolates the five hundred
$n=800$ replications (Section~\ref{sec:s6}).

All results in the paper derive from one logged run
(\texttt{results/log.txt}): Julia 1.13.0, single-threaded OpenBLAS,
Apple M2 Max processor, seed \texttt{MersenneTwister(18900217)}, the
date of birth of R.~A. Fisher, double precision, stopping tolerance
$\|\diag(e^{A[x]})-\one\|_{\infty}\leq10^{-13}$. To reproduce: run
\texttt{julia -t 1 runtests.jl}, then \texttt{julia -t 1 overnight.jl}
(four to six hours), then \texttt{julia -t 1 final\_checks.jl}, and
finally execute \texttt{figures.ipynb} with a Julia kernel and
Plots.jl.

\section{Numerically safe evaluation and algorithmic constants}\label{sec:s2}

This section records the floating-point detail behind the evaluation,
without overflow or cancellation, of the quantities that the GFT-FP+N
algorithm computes from the eigendecomposition, together with the
constants of its line search. Accuracy of the returned matrix is a
separate matter, governed by its conditioning as discussed in
Section~\ref{sec:numerics}.

\emph{The vector $\ell$.} With $A[x]=Q\Lambda Q'$ and $Q=(q_{ik})$,
the $i$th diagonal element of $e^{A[x]}$ is
$\sum_k q_{ik}^{2}e^{\lambda_k}$, which can overflow double precision
far from the solution. The algorithm therefore never forms it there,
computing instead the row-wise weighted log-sum-exp
\[
\ell_i=m_i+\log\Big(\sum\nolimits_k e^{a_{ik}-m_i}\Big),
\qquad
a_{ik}=\lambda_k+2\log|q_{ik}|,
\quad m_i=\max\nolimits_k a_{ik},
\]
with $\log0=-\infty$, in which every exponential is at most one. The
row-specific shift matters: a global shift $m=\max_k\lambda_k$ can
underflow to $-\infty$ for a row that is nearly orthogonal to the top
eigenspace of a very spread spectrum. The fixed-point update
$x\leftarrow x-\ell$ and the normalization \eqref{eq:norm} use only
$\ell$ and the eigenvalues; the algorithm takes fixed-point steps
while $\min_i\ell_i<-700$, a conservative threshold below which
$e^{\ell_i}$, though still representable
($e^{-700}\approx10^{-304}$), approaches the underflow limit of
double precision. The Newton system is formed in the scaled variables
$e^{-c}D$, $e^{-c}L$, and $e^{-c}D\ell$ with $c=\lambda_{\max}(A[x])$,
so that no intermediate quantity overflows even without
normalization.

\emph{The residual $g$.} In the Newton phase,
$g_i=e^{\ell_i}-1$ is evaluated as $\mathrm{expm1}(\ell_i)$, where
$\mathrm{expm1}(t)$ denotes the standard library function computing
$e^{t}-1$ to full relative accuracy for small $t$. The naive
evaluation retains only about three significant digits at
$\ell\approx10^{-13}$ and none near $10^{-16}$, which would
corrupt the convergence test
$\|g\|_{\infty}<10^{-13}$ at exactly the scale where it is applied.
The preconditioner diagonal is $d_i=e^{\ell_i}$, evaluated directly:
the algebraically equal $1+\mathrm{expm1}(\ell_i)$ cancels to zero
for $\ell_i\leq-37$.

\emph{The divided differences $L$.} The matrix $L$ of
Section~\ref{sec:algorithm} is
evaluated in the symmetric form
\[
L_{kl}=e^{m_{kl}}\,\frac{-\mathrm{expm1}(-r_{kl})}{r_{kl}},
\qquad
m_{kl}=\max(\lambda_k,\lambda_l),\quad r_{kl}=|\lambda_k-\lambda_l|,
\]
with the continuous value $e^{m_{kl}}$ at $r_{kl}=0$. Unlike one-sided
forms such as $e^{\lambda_l}\,\mathrm{expm1}(r)/r$, this cannot
produce $0\times\infty$ for widely spread spectra, and it is accurate
for coincident and nearly coincident eigenvalues.

\emph{The objective $f$ and the line search.} Every evaluated point
is normalized by \eqref{eq:norm}, with $s$ computed as a log-sum-exp
of the eigenvalues, so $\tr(e^{A[x]})=n$ and
$f=\sum_k e^{\lambda_k}-\sum_i x_i$ is computed from the eigenvalues
without overflow. The Armijo backtracking uses sufficient-decrease
constant $10^{-4}$, step halving from $\min(1,2t_{\mathrm{prev}})$,
and minimum step $10^{-8}$; trial points are normalized before $f$ is
evaluated, and trial points at which $f$ evaluates nonfinite are
rejected. If conjugate gradients reach the cap of $2n$ products before
the forcing condition holds, the current iterate is used. A truncated
direction satisfies $(D\ell)'\delta<0$ but not necessarily
$g'\delta<0$, and step~5 tests the latter before anything else.

\emph{The two safeguards.} Near the minimum the predicted decrease
$|g'\delta|$ falls below what floating-point comparison of $f$ values
can certify; when it falls below the conservative threshold
$10^{-12}(1+|f|)$, and $\|\ell\|_{\infty}<10^{-3}$ confirms that the
iterate is near the solution (a small $|g'\delta|$ alone does not,
for the system $H\delta=-D\ell$), the full Newton step is taken
without a test. At
the rounding floor of $\diag(e^{A[x]})$, of order
$\varepsilon_{\mathrm{mach}}\lambda_{\max}(e^{A[x]})$, the computed
gradient is dominated by rounding error and Newton directions cease to
make progress, while the fixed-point update $x\leftarrow x-\ell$
remains effective; the algorithm therefore switches permanently to
fixed-point steps once $\|g\|_{\infty}<10^{-9}$ fails to halve over
three iterations. Section~\ref{sec:s6} documents the replication that
motivated this safeguard.

\section{Cross-language agreement on serialized inputs}\label{sec:s3}

The test suite embeds golden values: input vectors $z$ and solutions
$x^{\ast}(z)$ stored with seventeen significant digits, generated by
an independent NumPy
implementation, developed and verified separately from the Julia code,
for a Toeplitz design, an equicorrelation design with repeated
eigenvalues, and a fixed low-dimensional $z$ with its reconstructed
correlation matrix. A passing run of \texttt{runtests.jl} verifies that
the Julia implementation reproduces these solutions to at most
$10^{-10}$ in the maximum norm, together with round-trip identities for
all four solvers, finite-difference checks of the gradient, agreement
of Hessian-vector products with the exact Hessian, and the bounds of
Theorem~\ref{prop:spec} at randomly drawn points; the R package's
test suite reproduces the same golden values. On the deterministic
designs of
Table~\ref{tab:bench}, the iteration counts of GFT-FP, Broyden's method,
and GFT-FP+N
agree exactly across the two implementations. On serialized random
inputs, twenty-five draws each of the Wishart, one-factor, and
$z\sim N(0,16I_d)$ designs shipped with the package
(\texttt{check\_serialized.jl}), the counts of Broyden's method and
GFT-FP+N agreed in all $75$ replications per method; the fixed-point
counts, which reach a thousand iterations, agreed in $44$ and
differed by at most $8$; the counts of Newton with Armijo
backtracking differed in $29$ replications, in one case by $573$, with $5$
disagreements in convergence status, consistent with the sensitivity
described next.

The exception is Newton with Armijo backtracking on near-singular
deterministic designs: its count on the Toeplitz $\rho=0{\cdot}99$,
$n=100$ problem was $7$ in NumPy and $10$, $23$ and $34$ in three
algebraically equivalent Julia builds, and $7$ in the build behind
Table~\ref{tab:bench}, whose failure counts on the random designs also
differ from an earlier build's by up to $20$ of $1{,}000$. The
accept--reject decisions of its Armijo line search sit near boundaries
at tolerance $10^{-13}$ and flip under rounding-level perturbations,
and the flips cascade. The counts of the other methods are stable;
this sensitivity is specific to an objective-based line search at
tight tolerance without the untested-step rule, and is one of the
reasons GFT-FP+N takes the full step untested within the numerical
safeguard described in Section~\ref{sec:s2}.

\section{The $n=800$ replications and the rounding floor}\label{sec:s6}

The terminal safeguard of the algorithm is motivated by the $n=800$
design, on which a stopping rule on
$\|\diag(e^{A[x]})-\one\|_{\infty}$ at tolerance $10^{-13}$ operates
at the rounding floor. Rounding errors in $\diag(e^{A[x]})$ scale with
$n$, and on this design the criterion evaluated at a reference
solution, $\diag(e^{A})$ from the eigendecomposition of $\log C$
itself, is of order $10^{-13}$ (up to $2{\cdot}1\times10^{-13}$ on the
draw with the largest condition number, $6{\cdot}4\times10^{4}$). With
the terminal safeguard disabled, a replication encountered during
development does not reach $10^{-13}$: its residual stalls at
$3{\cdot}6\times10^{-13}$, where Newton directions are computed from
gradients dominated by rounding error while fixed-point steps remain
effective. With the safeguard, all $500$ replications converge with
the algorithm of Section~\ref{sec:algorithm}, with median $8$ and
quartiles $7$ and $8$ eigendecompositions, but the floor is visible in
the tail: $16$ replications need between $21$ and $159$, spent in
terminal fixed-point mode with the residual wandering between
$10^{-13}$ and $3\times10^{-12}$ until an iterate falls below the
tolerance. The first-draft version of the algorithm, which needs
three more eigendecompositions on this design, exhausted the cap of
$500$ iterations in that mode on one of the $500$ draws, with
residuals between $4\times10^{-13}$ and $1{\cdot}6\times10^{-12}$
(\texttt{results\_v110\_20260918}); GFT-FP reached the tolerance on
the same draw. At tolerance $10^{-12}$, appropriate at this size, all
$500$ replications converge in at most $14$ eigendecompositions,
median $7$, with no terminal fixed-point steps
(\texttt{recheck\_n800.jl}, which takes the tolerance as an argument
and replays the same draws). The counts of Table~\ref{tab:bench} for
this design, reported at $10^{-13}$ for comparability with the other
rows, are therefore conservative for GFT-FP+N.

\section{Variants, safeguarded acceleration, and sequential inversion}\label{sec:variants}

\emph{Ablation of the Newton step.} GFT-FP+N of
Section~\ref{sec:algorithm} differs from the version in the first
draft of this paper in four ways. It solves $H\delta=-D\ell$ instead
of $H\delta=-g$. It uses the forcing tolerance $\min(1/2,\|\ell\|)$
instead of the square root. It normalizes every evaluated point by
\eqref{eq:norm} and drops the fixed-point phase while
$\max_i\ell_i>\log2$. And it starts the Armijo search from
$\min(1,2t_{\mathrm{prev}})$. Table~\ref{tab:ablation} varies these
choices one at a time on the draws of Table~\ref{tab:bench}, with all
other rules of Section~\ref{sec:algorithm} unchanged, including the
$-700$ rule, terminal mode, the untested-step rule, the constants of
Section~\ref{sec:s2}, and the caps. On the $z$-designs the log residual and quadratic forcing each save
about two eigendecompositions under otherwise identical rules, and the
effects add: D needs $13$ and $18$ against $17$ and $23$ for A.
Normalization alone changes little. Dropping the phase separates the
two residuals sharply: it costs the gradient residual two to seven
eigendecompositions (AN0, BN0) and saves the log residual three to
four (CN0, DN0). The adaptive initial step trims the tail ($\sigma=4$
maximum $28$ to $26$) without moving the medians. Products stay between
$57$ and $95$ throughout, so the savings in eigendecompositions are
savings in time: the median wall time falls by $27$\% and $28$\% from A
to DN0t. On the structured designs the same ordering holds, with the
final variant needing $5$, $6$, $6$, $6$, $5$, and $7$ eigendecompositions
on the six designs of Table~\ref{tab:bench} with $n\leq300$ against
$7$, $8$, $9$, $9$, $7$, and $10$ for A, and $8$ against $11$ (median
over the first $50$ draws) at $n=800$, where an eigendecomposition
costs only about three products and the time saving shrinks to
$13$\%. Products matter more as $n$ grows: on the Toeplitz design at
$n=300$, quadratic forcing with the gradient residual (B) needs $44$
products against $26$ for A and is slower despite two fewer
eigendecompositions, whereas with the log residual (D) the count
stays at $26$. No variant failed or produced a non-descent direction
on any design.

\begin{table}[h]
\centering
\caption{Ablation of the Newton step on the $z$-designs, $n=50$,
$1{,}000$ draws each: median and maximum of the number of
eigendecompositions $E$, median number of Hessian--vector products
$P$, and median wall time in milliseconds. The residual is $g$ or
$\ell$; $a$ is the exponent in $\eta=\min(1/2,\|\cdot\|^{a})$; N marks
normalization of every evaluated point, phase the fixed-point phase
while $\max_i\ell_i>\log2$, and t the adaptive initial step. No
variant failed or produced a non-descent direction.}\label{tab:ablation}
\medskip
\small
\begin{tabular}{llcccc|ccc|ccc}
\toprule
 & & & & & & \multicolumn{3}{c|}{$\sigma=2$} & \multicolumn{3}{c}{$\sigma=4$}\\
Label & residual & $a$ & N & phase & t & $E$ [max] & $P$ & ms & $E$ [max] & $P$ & ms\\
\midrule
A (first draft) & $g$ & $1/2$ &  & \checkmark &  & 17 [23] & 64 & 4.4 & 23 [51] & 87 & 5.9\\
B & $g$ & $1$ &  & \checkmark &  & 15 [21] & 71 & 4.0 & 21 [48] & 88 & 5.4\\
C & $\ell$ & $1/2$ &  & \checkmark &  & 15 [20] & 61 & 4.0 & 21 [36] & 86 & 5.3\\
D & $\ell$ & $1$ &  & \checkmark &  & 13 [17] & 63 & 3.5 & 18 [33] & 84 & 4.9\\
\addlinespace[2pt]
AN & $g$ & $1/2$ & \checkmark & \checkmark &  & 16 [21] & 57 & 4.2 & 22 [36] & 78 & 5.6\\
BN & $g$ & $1$ & \checkmark & \checkmark &  & 14 [18] & 58 & 3.7 & 20 [33] & 88 & 5.2\\
CN & $\ell$ & $1/2$ & \checkmark & \checkmark &  & 15 [19] & 59 & 3.9 & 21 [35] & 80 & 5.4\\
DN & $\ell$ & $1$ & \checkmark & \checkmark &  & 13 [17] & 66 & 3.6 & 19 [33] & 79 & 4.9\\
\addlinespace[2pt]
AN0 & $g$ & $1/2$ & \checkmark &  &  & 19 [41] & 62 & 4.9 & 30 [63] & 90 & 7.4\\
BN0 & $g$ & $1$ & \checkmark &  &  & 17 [38] & 69 & 4.5 & 28 [60] & 95 & 6.9\\
CN0 & $\ell$ & $1/2$ & \checkmark &  &  & 13 [20] & 63 & 3.6 & 17 [31] & 92 & 4.9\\
DN0 & $\ell$ & $1$ & \checkmark &  &  & 10 [18] & 68 & 3.2 & 15 [28] & 95 & 4.6\\
\addlinespace[2pt]
CN0t & $\ell$ & $1/2$ & \checkmark &  & \checkmark & 13 [19] & 63 & 3.6 & 16 [29] & 92 & 4.7\\
DN0t (Section~4) & $\ell$ & $1$ & \checkmark &  & \checkmark & 10 [16] & 68 & 3.2 & 14 [26] & 93 & 4.3\\
\bottomrule
\end{tabular}
\end{table}

\emph{Descent.} The system $H\delta=-D\ell$ need not produce a
descent direction for $f$. For the $3\times3$ matrix $A$ with
off-diagonal elements $A_{12}=4{\cdot}5$, $A_{13}=11{\cdot}5$,
$A_{23}=-12$ and diagonal $x=(-3{\cdot}9,-16{\cdot}1,-25{\cdot}5)'$, one
has $\ell=(0{\cdot}916,-4{\cdot}758,-0{\cdot}644)'$ and
$g'\delta=9{\cdot}46>0$ for the exact solution of $H\delta=-D\ell$,
with $\operatorname{cond}(H)=2{\cdot}0\times10^{3}$. Step~5 therefore
tests $g'\delta<0$ before the untested-step rule. In all our runs the
test never failed. Backtracking on $\psi(x)=\tfrac12\|\ell(x)\|^{2}$
instead of $f$ is an overflow-free alternative for which
$H\delta=-D\ell$ is always a descent direction when the forcing test
holds, since $\nabla\psi'\delta=\ell'D^{-1}H\delta\leq-(1-\eta)\|\ell\|^{2}$;
it gives a globally convergent inexact Newton method in the sense of
\citet{EisenstatWalker:1994}, but in development runs it needed about
one more eigendecomposition than backtracking on $f$.

\emph{Safeguarded acceleration.} Let $T(x)=x-\ell(x)$ and
$V(x)=\sum_i(e^{\ell_i(x)}-1-\ell_i(x))$, and fix $\theta\in(0,1)$.
From any $x_0$, let $x_{k+1}=y_k$ if $f(y_k)\leq f(x_k)-\theta V(x_k)$
and $x_{k+1}=T(x_k)$ otherwise, where the finite trial points $y_k$
are arbitrary. Then $x_k\to x^{\ast}(z)$. Indeed, by
Proposition~\ref{prop:descent} both branches give
$f(x_{k+1})\leq f(x_k)-\theta V(x_k)$, so
$\theta\sum_kV(x_k)\leq f(x_0)-f(x^{\ast})<\infty$. The iterates stay
in the compact sublevel set $\{f\leq f(x_0)\}$, and since $e^{t}-1-t>0$
for $t\neq0$, $V(x_k)\to0$ forces $\ell(x_k)\to0$. Every cluster point
therefore solves $\ell=0$, so equals $x^{\ast}(z)$, and compactness
gives convergence of the whole sequence. No bound on mixing
coefficients is needed. Near $x^{\ast}$, a full Newton step of
Section~\ref{sec:algorithm} decreases $f$ by
$\tfrac12\ell'H^{\ast-1}\ell+O(\|\ell\|^{3})$ and
$V=\tfrac12\|\ell\|^{2}+O(\|\ell\|^{3})$. Since $H^{\ast}\preceq I$,
the ratio of the two has lower limit at least one (numerically about
two on the $z$-designs), so every $\theta<1$ eventually accepts full
steps. In floating point, neither $f$-differences nor small values of
$V$ can be certified near the rounding floor. The implementation
therefore accepts trial points untested once
$\theta V(x)\leq10^{-12}(1+|f|)$, as in step~5, and evaluates
$e^{t}-1-t$ by its series for $|t|<10^{-3}$. The convergence statement
concerns exact arithmetic. Used as the only globalization of the
Newton phase, one full trial step and a fixed-point step on rejection,
the safeguard was slower than Armijo backtracking in development runs,
since far from $x^{\ast}$ a rejected trial wastes an eigendecomposition
and the fixed-point step then makes less progress than a damped Newton
step. It is therefore used for Anderson acceleration, which has no
line search.

\emph{Anderson acceleration: memory and safeguard.}
Table~\ref{tab:aasweep} repeats the comparison of
Section~\ref{sec:tools} for memory $5$ to $50$, with and without the
safeguard ($\theta=1/4$), on the draws of Table~\ref{tab:bench}.
More memory saves eigendecompositions up to memory $20$ to $50$ and
time up to memory $10$ to $20$; beyond that the least-squares solve
offsets the saving and the tail grows. Without the safeguard, memory
$20$ diverged on $2$ and memory $50$ on $9$ of the $1{,}000$ draws of
the $\sigma=4$ design; the iterates never reached $\|\ell\|_{\infty}<1$
and drifted away. With the safeguard every draw converged. At memory
$5$ and $10$ the counts are unchanged to within one eigendecomposition
and the time within one percent; at memory $50$ the guarded medians
are $60$ and $83$ against $60$ and $76$ over the converging draws of
the unguarded scheme. After a rejection the memory is cleared; this
is not needed for convergence, but it avoids reusing the differences
that produced the rejected proposal. The fastest Anderson
configurations remain about four times slower than GFT-FP+N.

\begin{table}[h]
\centering
\caption{Anderson acceleration on the $z$-designs, $n=50$, $1{,}000$
draws: median [maximum] of the number of eigendecompositions over
converging draws, failures as superscripts, and median wall time in
milliseconds.}\label{tab:aasweep}
\medskip
\begin{tabular}{lcc|cc}
\toprule
 & \multicolumn{2}{c|}{$\sigma=2$} & \multicolumn{2}{c}{$\sigma=4$}\\
Method & $E$ [max] & ms & $E$ [max] & ms\\
\midrule
GFT-FP+N & 10 [16] & 3.2 & 14 [26] & 4.3\\
Anderson, $m=5$ & 93 [133] & 17.9 & 156 [241] & 30.2\\
Anderson, $m=5$, guarded & 93 [133] & 17.8 & 158 [243] & 30.4\\
Anderson, $m=10$ & 72 [104] & 14.5 & 112 [182] & 22.6\\
Anderson, $m=10$, guarded & 72 [97] & 14.5 & 112 [173] & 22.5\\
Anderson, $m=20$ & 63 [79] & 14.4 & 97 [209]$^{2}$ & 22.7\\
Anderson, $m=20$, guarded & 63 [76] & 14.2 & 95 [118] & 21.5\\
Anderson, $m=50$ & 60 [109] & 17.5 & 76 [225]$^{9}$ & 24.4\\
Anderson, $m=50$, guarded & 60 [109] & 16.9 & 83 [293] & 21.8\\
\bottomrule
\end{tabular}
\end{table}

\emph{Sequential inversion.} Table~\ref{tab:warmpred} repeats the
warm-start experiment of Section~\ref{sec:numerics} for increment
standard deviations $0{\cdot}0005$, $0{\cdot}002$, and $0{\cdot}008$,
for the first-draft and the present version of GFT-FP+N, with and
without the predictor \eqref{eq:pred} solved to relative residual
$10^{-8}$ or $10^{-3}$. Averages include the cold start and the
predictor's products and times. The predictor saves at least one eigendecomposition per step at every
step size, and with the algorithm of Section~\ref{sec:algorithm} it
brings the count to exactly three per warm step: the evaluation at
the predicted point and one after each of two Newton steps. The
relative tolerance $10^{-3}$ is as accurate as $10^{-8}$ at these step
sizes while saving seven to eight products per step. The time gain,
however, depends on the step: at the step size of
Section~\ref{sec:numerics} it is $23$\%, at a quarter of that step
warm starts alone are nearly as good, and at four times that step the
predictor's conjugate-gradient solve costs more than the
eigendecomposition it saves, so the time rises by $8$\%. Because the
cost of an eigendecomposition relative to a product falls with $n$,
the gain should be expected to shrink at larger dimensions. (The
first-draft rows and the warm-start row of Section~\ref{sec:numerics}
come from separate runs and differ at the level of timing noise.)

\begin{table}[h]
\centering
\caption{Warm starts and the tangent predictor: means per step over
ten sequences of $50$ inversions, $n=100$, including the cold start.
$E$ is the number of eigendecompositions, $P$ the number of
Hessian--vector products, and $B\Delta z$ the number of directional
derivatives \eqref{eq:pred}.}\label{tab:warmpred}
\medskip
\small
\begin{tabular}{llcccc}
\toprule
sd of increments & Method & $E$/step & $P$/step & $B\Delta z$/step & ms/step\\
\midrule
$0{\cdot}0005$ & first draft, warm start & 5.73 & 24.8 & 0.00 & 9.56\\
 & first draft, predictor, $10^{-8}$ & 4.13 & 33.8 & 0.98 & 8.37\\
 & first draft, predictor, $10^{-3}$ & 4.13 & 26.3 & 0.98 & 7.61\\
 & Section~4, warm start & 4.07 & 24.7 & 0.00 & 7.05\\
 & Section~4, predictor, $10^{-8}$ & 3.08 & 38.1 & 0.98 & 7.59\\
 & Section~4, predictor, $10^{-3}$ & 3.08 & 30.6 & 0.98 & 6.90\\
\addlinespace[2pt]
$0{\cdot}002$ & first draft, warm start & 6.10 & 23.4 & 0.00 & 9.50\\
 & first draft, predictor, $10^{-8}$ & 4.51 & 33.8 & 0.98 & 8.82\\
 & first draft, predictor, $10^{-3}$ & 4.51 & 26.4 & 0.98 & 8.05\\
 & Section~4, warm start & 4.81 & 32.8 & 0.00 & 8.75\\
 & Section~4, predictor, $10^{-8}$ & 3.08 & 33.8 & 0.98 & 7.22\\
 & Section~4, predictor, $10^{-3}$ & 3.08 & 26.2 & 0.98 & 6.73\\
\addlinespace[2pt]
$0{\cdot}008$ & first draft, warm start & 7.28 & 27.4 & 0.00 & 9.30\\
 & first draft, predictor, $10^{-8}$ & 5.38 & 37.0 & 0.98 & 8.94\\
 & first draft, predictor, $10^{-3}$ & 5.39 & 29.2 & 0.98 & 8.16\\
 & Section~4, warm start & 5.06 & 26.7 & 0.00 & 7.33\\
 & Section~4, predictor, $10^{-8}$ & 4.06 & 46.6 & 0.98 & 8.82\\
 & Section~4, predictor, $10^{-3}$ & 4.06 & 38.8 & 0.98 & 7.93\\
\bottomrule
\end{tabular}
\end{table}

\emph{Options not adopted.} Scale can also be eliminated exactly:
writing $x=y+c\one$ with $\one'y=0$ and minimizing over $c$ gives
$F(y)=n+n\log\{\tr(e^{A[y]})/n\}$, whose Hessian at a normalized point
is $H-dd'/n\succeq0$ with null space $\mathrm{span}(\one)$. We found
no speed advantage over normalization and did not pursue it; we note
that $F$ has a globally Lipschitz gradient with sharp constant $n/2$,
since $H-dd'/n\preceq D-dd'/n\preceq(n/2)I$ at a normalized point,
which places the reduced problem in the smooth convex class.

\section{Quadrature preconditioning}\label{sec:quad}

\emph{Construction.} For $A=A[x]=Q\Lambda Q'$ and the $r$-point
Gauss--Legendre rule $(t_j,a_j)$ on $[0,1]$, let
$M_r=\sum_ja_je^{t_jA}\circ e^{(1-t_j)A}$ as in \eqref{eq:Mr}. Each
$e^{tA}=Q\diag(e^{t\lambda})Q'$ costs one matrix multiplication, and a
symmetric rule pairs $t_j$ with $1-t_j$, so $r$ multiplications and
one Cholesky factorization suffice; with the scaling $e^{-\lambda_{\max}}$
used for \eqref{eq:hv} all entries are bounded. The factor replaces
$D$ inside the conjugate-gradient recurrence, at two triangular solves
per iteration, and the forcing test of step~4 is unchanged. A failed
factorization falls back to $D$.

\emph{Certificate.} With $W=Q'\diag(v)Q$, \eqref{eq:quadform} gives
$v'Hv=\sum_{a,b}L_{ab}W_{ab}^{2}$ and
$v'M_rv=\sum_{a,b}L^{(r)}_{ab}W_{ab}^{2}$, where
$L_{ab}=\int_0^1e^{t\lambda_a+(1-t)\lambda_b}dt$ and $L^{(r)}_{ab}$ is
its quadrature approximation; since $W_{ab}^{2}\geq0$, bounds on the
coefficient ratios are inherited by the matrices. Both coefficients
carry the factor $e^{(\lambda_a+\lambda_b)/2}$, so for a symmetric rule
the ratio is $\phi_r(\lambda_a-\lambda_b)$ with
\[
\phi_r(u)=\frac{\sinh(u/2)/(u/2)}{\sum_ja_j\cosh\{(t_j-\tfrac12)u\}} .
\]
For Gauss--Legendre nodes the quadrature error for $e^{tu}$ is a
positive multiple of $u^{2r}e^{\xi u}$ \citep{DavisRabinowitz:1984},
so $\phi_r\geq1$ with equality at $u=0$, which occurs on the diagonal.
For $r=1$ the rule is the midpoint rule and $M_1=e^{A/2}\circ e^{A/2}$,
so the lower bound specializes to $e^{A/2}\circ e^{A/2}\preceq H$, a
bound established in \citet{ArchakovHansen:Legendre} together with
the sharp relative bound $\alpha\{\ell(E)\}D\preceq H\preceq D$ in
terms of the spread $\ell(E)$ of the spectrum of $E=e^{A[x]}$; the
quadrature family refines the lower bound with $r$ and certifies the
ratio.
Hence $M_r\preceq H\preceq\beta_r(\Delta)M_r$ with
$\beta_r(\Delta)=\sup_{|u|\leq\Delta}\phi_r(u)$ and
$\Delta=\lambda_{\max}-\lambda_{\min}$; numerically $\phi_r$ increases
in $|u|$, so $\beta_r(\Delta)=\phi_r(\Delta)$. Two nodes certify
$\kappa\leq2$ for $\Delta\leq11{\cdot}6$, three for
$\Delta\leq22{\cdot}4$, and four for $\Delta\leq36{\cdot}7$. At the
solution $\Delta=\log\operatorname{cond}(C)$, so the order is set by
the conditioning of $C$: one node for the Toeplitz $\rho=0{\cdot}5$ and
Wishart designs, two for Toeplitz $\rho=0{\cdot}9$ and $0{\cdot}99$
and the one-factor designs, and five to eight for the $z$-designs.

\emph{Selection rule and evidence.} The rule uses $M_r$ at a Newton
step when $\kappa\leq2$ is certified with two to eight nodes, and $D$
when a single node suffices or none of these orders certifies the
bound. Table~\ref{tab:quad} reports the diagonal
preconditioner, quadrature at every step, and the rule on the designs
of Table~\ref{tab:bench} and on larger one-factor, Wishart, and
Toeplitz designs. In Julia the setup is cheap relative to the products, and quadrature
at every step gains on every design where two or more nodes are
certified: the one-factor designs gain $1{\cdot}2$ at $n=100$,
$1{\cdot}3$ at $n=200$ and $400$, and $1{\cdot}5$ at $n=800$, always
with one eigendecomposition fewer and half the products; Toeplitz
$\rho=0{\cdot}9$ gains $1{\cdot}4$ at $n=100$ and $1{\cdot}9$ at
$n=800$; the $z$-designs gain $1{\cdot}2$ despite five to eight nodes
per build, since their products fall from $68$ and $93$ to $17$ and
$19$; Toeplitz $\rho=0{\cdot}99$ is neutral. It loses only where a
single node certifies the bound, Toeplitz $\rho=0{\cdot}5$ and Wishart
at $n=400$ (ratios $0{\cdot}9$), because the diagonal solve is
already fast there. The rule keeps the diagonal preconditioner in
exactly those cases and matches quadrature elsewhere, so it never
loses in the table. No factorization failed and no run failed to
converge.

\begin{table}[h]
\centering
\caption{Quadrature preconditioning: median eigendecompositions $E$ and
Hessian--vector products $P$, median number of nodes per build, and
median wall time relative to the diagonal preconditioner (ratio $>1$
is faster), tolerance $10^{-13}$; one-factor designs marked
$\dagger$ use loadings $\mathrm{Un}(0{\cdot}8,0{\cdot}995)$. No run
failed.}\label{tab:quad}
\medskip
\small
\begin{tabular}{lrcc|ccc|cc}
\toprule
 & & \multicolumn{2}{c|}{Diagonal $D$} & \multicolumn{3}{c|}{$M_r$ at every step} & \multicolumn{2}{c}{Rule}\\
Design & draws & $E\{P\}$ & ms & $E\{P\}$ & $r$ & ratio & $E\{P\}$ & ratio\\
\midrule
Toeplitz $\rho=0{\cdot}5$, $n=100$ & 1 & 5\{15\} & 5.5 & 5\{15\} & 1 & 0.91 & 5\{15\} & 1.00\\
Toeplitz $\rho=0{\cdot}9$, $n=100$ & 1 & 6\{30\} & 7.7 & 5\{8\} & 2 & 1.43 & 5\{8\} & 1.43\\
Toeplitz $\rho=0{\cdot}99$, $n=100$ & 1 & 6\{29\} & 7.6 & 6\{20\} & 2 & 1.00 & 6\{20\} & 1.01\\
Toeplitz $\rho=0{\cdot}99$, $n=300$ & 1 & 6\{28\} & 109 & 6\{21\} & 2 & 1.01 & 6\{21\} & 1.01\\
Wishart, $n=100$ & 1000 & 5\{15\} & 5.6 & 4\{7\} & 1 & 1.33 & 5\{15\} & 1.00\\
One-factor, $n=100$ & 1000 & 7\{34\} & 11.0 & 6\{16\} & 2 & 1.23 & 6\{16\} & 1.23\\
$z\sim N(0,4I_d)$, $n=50$ & 1000 & 10\{68\} & 3.2 & 10\{17\} & 5 & 1.20 & 10\{17\} & 1.21\\
$z\sim N(0,16I_d)$, $n=50$ & 1000 & 14\{93\} & 4.3 & 12\{19\} & 8 & 1.16 & 12\{19\} & 1.17\\
One-factor, $n=800$ & 500 & 8\{35\} & 2305 & 6\{16\} & 2 & 1.46 & 6\{16\} & 1.44\\
\addlinespace[2pt]
One-factor$^{\dagger}$, $n=200$ & 1000 & 7\{34\} & 54.3 & 6\{16\} & 2 & 1.27 & 6\{16\} & 1.27\\
One-factor$^{\dagger}$, $n=400$ & 200 & 7\{32\} & 324 & 6\{16\} & 2 & 1.30 & 6\{16\} & 1.30\\
Wishart, $n=400$ & 100 & 4\{7\} & 99.4 & 4\{7\} & 1 & 0.90 & 4\{7\} & 1.00\\
Toeplitz $\rho=0{\cdot}5$, $n=800$ & 1 & 5\{15\} & 1059 & 5\{15\} & 1 & 0.91 & 5\{15\} & 1.00\\
Toeplitz $\rho=0{\cdot}9$, $n=800$ & 1 & 6\{31\} & 1784 & 5\{8\} & 2 & 1.85 & 5\{8\} & 1.85\\
Toeplitz $\rho=0{\cdot}99$, $n=800$ & 1 & 6\{30\} & 1747 & 6\{21\} & 2 & 1.06 & 6\{21\} & 1.06\\
\bottomrule
\end{tabular}
\end{table}

\emph{Sequences.} In the warm-start design of
Section~\ref{sec:numerics} with the tangent predictor, at $n=100$,
$400$, and $800$ (aggregate step $0{\cdot}14$ per step at every
dimension), the predictor already leaves about three
eigendecompositions per step, so quadrature acts on the products:
with the diagonal preconditioner a warm step costs $3{\cdot}08$, $3{\cdot}20$, and $3{\cdot}83$
eigendecompositions and 26, 30, and 39 products at the three
dimensions, and with $M_r$ rebuilt at every Newton step $3{\cdot}06$, $3{\cdot}17$,
and $3{\cdot}60$ eigendecompositions and 18, 18, and 18 products; the
time per step falls by 8\%, 18\%, and 33\% (from $6{\cdot}5$ to $6{\cdot}0$
milliseconds at $n=100$, 257 to 212 at $n=400$, and $2{\cdot}2$ to $1{\cdot}5$
seconds at $n=800$). The predictor solve uses the same factor.

\bibliographystyle{apalike}
\bibliography{GFTarxiv}

\end{document}